\documentclass[
final
]{dmtcs-episciences}

\usepackage[utf8]{inputenc}
\usepackage[T1]{fontenc}
\usepackage{enumerate}
\usepackage{amsmath}
\usepackage{amsfonts}
\usepackage{nicefrac}
\usepackage{amssymb}
\usepackage{amsthm}
\usepackage{mathtools}
\usepackage{tikz}
\usetikzlibrary{calc}
\usepackage{pgfplots}
\pgfplotsset{compat=newest}
\usepackage{pgfplotstable}
\usepackage{xcolor}
\usepackage{subfigure}
\usepackage{microtype}
\usepackage{tabularx}
\usepackage{booktabs}
\usepackage{paralist}

\usepackage[ruled,vlined,linesnumbered]{algorithm2e}
\DontPrintSemicolon

\usepackage[round,sort&compress]{natbib}
\usepackage{hyperref}
\usepackage[nameinlink,sort&compress,capitalize]{cleveref}

\newtheorem{theorem}{Theorem}
\newtheorem{lemma}[theorem]{Lemma}
\newtheorem{observation}[theorem]{Observation}
\newtheorem{corollary}[theorem]{Corollary}
\newtheorem{proposition}[theorem]{Proposition}

\theoremstyle{definition}
\newtheorem{definition}[theorem]{Definition}

\newtheorem{rrule}{Rule}

\crefname{rrule}{Rule}{Rules}
\crefname{case}{Case}{Cases}

\newcommand{\prob}[1]{\textnormal{\textsc{#1}}}

\newcommand{\problemdef}[3]{
	\begin{center}
	\begin{minipage}{0.95\columnwidth}
		\noindent
		\prob{#1}
		\vspace{5pt}\\
		\setlength{\tabcolsep}{3pt}
		\begin{tabularx}{\textwidth}{@{}lX@{}}
			\textbf{Input:}     & #2 \\
			\textbf{Question:}  & #3
		\end{tabularx}
	\end{minipage}
	\end{center}
}

\DeclarePairedDelimiterX{\abs}[1]{\lvert}{\rvert}{#1}
\DeclarePairedDelimiterX{\norm}[1]{\lVert}{\rVert}{#1}
\DeclarePairedDelimiterX{\ceil}[1]{\lceil}{\rceil}{#1}

\newcommand{\NN}{\mathbb{N}}

\newcommand{\Wone}{\textnormal{\textrm{W[1]}}}
\newcommand{\coWone}{\textnormal{\textrm{coW[1]}}}
\newcommand{\Wtwo}{\textnormal{\textrm{W[2]}}}
\newcommand{\coWtwo}{\textnormal{\textrm{coW[2]}}}
\newcommand{\FPT}{\textnormal{\textrm{FPT}}}

\newcommand{\bigO}{\mathcal{O}}
\newcommand{\yes}{\textnormal{\texttt{yes}}}
\newcommand{\no}{\textnormal{\texttt{no}}}

\newcommand{\sside}{\ensuremath{\mathsf{ss}}}

\newcommand{\vc}{\mathrm{vc}}
\newcommand{\fvn}{\mathrm{fvn}}
\newcommand{\tw}{\mathrm{tw}}
\newcommand{\td}{\mathrm{td}}
\newcommand{\fen}{\mathrm{fen}}

\DeclareMathOperator{\opt}{opt}

\newcommand{\oneto}[1]{[ #1 ]} %

\newcommand{\ie}{i.\,e.,\ }

\newcommand{\BDD}{\prob{BDD}}
\newcommand{\BFVD}{\prob{BFVD}}

\author[Lito Goldmann, Leon Kellerhals, and Tomohiro Koana]{%
	Lito Goldmann\affiliationmark{1}
	\and Leon Kellerhals\affiliationmark{1}
	\and Tomohiro Koana\affiliationmark{2}\thanks{Work was done while affiliated with TU Berlin. Supported by the DFG Project DiPa, Ni 369/21.}
}
\title[Structural Parameterizations of Biclique-Free Deletion]{Structural Parameterizations of the Biclique-Free Vertex Deletion Problem}
\affiliation{
	Technische Universität Berlin, Germany\\
	Utrecht University, The Netherlands
}
\keywords{
	Fixed-parameter tractability,
	Kernelization,
	Structural graph parameterizations,
	Biclique-free graphs
}

\begin{document}
\publicationdata{vol. 26:3}{2024}{4}{10.46298/dmtcs.13018}{2024-02-07; 2024-02-07; 2024-07-16}{2024-08-30}

\maketitle

\begin{abstract}
	In this work, we study the Biclique-Free Vertex Deletion problem:
	Given a graph $G$ and integers $k$ and $i \le j$, find a set of at most $k$ vertices that intersects every (not necessarily induced) biclique $K_{i, j}$ in $G$.
	This is a natural generalization of the Bounded-Degree Deletion problem, wherein one asks whether there is a set of at most $k$ vertices whose deletion results in a graph of a given maximum degree $r$.
	The two problems coincide when $i = 1$ and $j = r + 1$.
	We show that Biclique-Free Vertex Deletion is fixed-parameter tractable with respect to $k + d$ for the degeneracy $d$ by developing a $2^{\bigO(d k^2)} \cdot n^{\bigO(1)}$-time algorithm.
	We also show that it can be solved in $2^{\bigO(f k)} \cdot n^{\bigO(1)}$ time for the feedback vertex number $f$ when $i \ge 2$.
	In contrast, we find that it is W[1]-hard for the treedepth for any integer $i \ge 1$.
	Finally, we show that Biclique-Free Vertex Deletion has a polynomial kernel for every $i \ge 1$ when parameterized by the feedback edge number.
	Previously, for this parameter, its fixed-parameter tractability for $i = 1$ was known (Betzler et al., 2012) but the existence of polynomial kernel was open.
\end{abstract}

\section{Introduction}

The \textsc{$\mathcal G$-Vertex Deletion} problem, which, for a graph class $\mathcal G$, asks whether a given graph $G$ can be turned into a graph $G' \in \mathcal G$ by deleting at most $k$ vertices,
is arguably one of the most pervasive and general graph theoretical problems.
In this work, we focus on the class of \emph{biclique-free} graphs, which has received considerable attention from algorithmic perspectives \citep{DBLP:journals/jcss/LokshtanovMPRS18,DBLP:journals/siamdm/EibenKMPS19,DBLP:journals/tcs/TelleV19,DBLP:conf/stacs/FabianskiPST19,DBLP:journals/algorithmica/AboulkerBKS23,DBLP:conf/stacs/KoanaKNS22}.
For $i, j \in \NN$, let $K_{i,j}$ denote the complete bipartite graph on $i$ vertices on one side and $j$ vertices on the other side.
We consider the following problem.

\problemdef{Biclique Free Vertex Deletion (\BFVD)}
	{An undirected graph~$G$ and $i,j,k \in \NN$, $i \le j$.}
	{Does there exist a subset~$V' \subseteq V$ with~$|V'|\leq k$ such that~$G - V'$ does not contain any~$K_{i,j}$ as a (not necessarily induced) subgraph?}

Note that we consider $i$ and $j$ to be part of the input, that is, they are not to be treated as a constant.
Hence, \BFVD{} is a generalization of the \textsc{Bounded Degree Deletion} problem, defined as follows.

\problemdef{Bounded Degree Deletion (\BDD{})}
	{An undirected graph~$G$ and $r, k \in \NN$}
	{Is there a subset~$V'\subseteq V$ with~$\abs{V'}\le k$ such that each vertex in~$G-V'$ has degree at most~$r$?}
Note that an instance~$(G,k,r)$ of \BDD{} is a \yes-instance, if and only if the instance~$(G,1,r+1,k)$ of \BFVD{} is a \yes-instance.
Furthermore, note that the special case~$r=0$ of~\BDD{}, that is,~\BFVD{} with~$i=j=1$, is~\textsc{Vertex Cover}.

\BDD{} appears in the field of computational biology~\citep{Fellows2011}.
Its dual, the so-called \textsc{$k$-Plex Deletion} problem, is a clique relaxation problem that finds many applications in social network analysis \citep{DBLP:journals/ior/BalasundaramBH11,DBLP:journals/jco/MoserNS12,DBLP:journals/jco/McCloskyH12,Seidman1978AGG}.
Hence, it is not surprising that its computational complexity has been studied extensively in the last two decades~\citep{DBLP:journals/ol/BalasundaramCT10,Betzler2012,DBLP:journals/iandc/BodlaenderF01,DBLP:conf/aaim/ChenFFJLWZ10,Dessmark1993,DBLP:journals/tcs/KomusiewiczHMN09,Nishimura2005,Seidman1978AGG}.
Its parameterized complexity has been studied as well: \BDD{} is fixed-parameter tractable (\FPT{}) with respect to $r+k$ \citep{Fellows2011,DBLP:journals/jco/MoserNS12,Nishimura2005}, but \Wtwo-hard with respect to~$k$ \citep{Fellows2011}.
As for structural parameterizations, \BDD{} is known to be \FPT{} with respect to the degeneracy plus $k$ \citep{DBLP:conf/cocoa/RamanSS08} and with respect to the feedback edge number and to the treewidth plus $r$ \citep{Betzler2012}.
Recently, \citet{Ganian2021} showed that \BDD{} is \Wone-hard when parameterized by the feedback vertex number or the treedepth,\footnote{see \cref{sec:prelim} for a definition of the parameters} but \FPT{} when parameterized by the treecut width \citep{DBLP:journals/siamdm/MarxW14}.
\citet{lampis2023structural} proved some fine-grained conditional lower bounds for \BDD{} with respect to treewidth and vertex cover number.

Allowing $i$ and $j$ to be part of the input makes \BFVD{} challenging to solve.
Deciding whether the input graph $G$ is free of bicliques $K_{i,j}$ is challenging on its own.
In fact, the problem of determining whether $G$ contains a biclique is NP-hard and \Wone-hard with respect to $i + j$.
Thus, \BFVD{} is coNP-hard and \coWtwo-hard for $i + j$ even if $k = 0$ \citep{DBLP:journals/jacm/Lin18}.
On degenerate graphs however, one can efficiently enumerate all maximal bicliques \citep{eppstein1994arboricity}.
For this reason, and in order to see which results for \BDD{} also hold for its generalization, we study the computational tractability of \BFVD{} with respect to structural graph parameters.

\paragraph{Our results.}

We first show in \Cref{sec:prelim} that \BFVD{} can be solved in $\bigO^*(2^{\bigO(\vc \cdot k)})$ time, where $\vc$ is the minimum vertex cover size of $G$.
This paves the way for the algorithms presented in \Cref{sec:d,sec:fvn}.
Using the $\bigO^*(2^d)$-time algorithm of \citet{eppstein1994arboricity}, where $d$ is the degeneracy of $G$, to enumerate all maximal bicliques, we show that each vertex and edge not part of any biclique $K_{i,j}$ can be identified (and deleted) in time $\bigO^*(4^d)$.
When every edge is part of some biclique $K_{i,j}$, the set of vertices that appear in the smaller side of some biclique $K_{i,j}$ form a vertex cover.
In \Cref{sec:d}, we show that \BFVD{} can be solved in $\bigO^*(2^{\bigO(dk^2)})$ time.
The algorithm takes a win-win approach:
If there are not many vertices that appear in the smaller side of a biclique, then we use the aforementioned $\bigO^*(2^{\bigO(\vc \cdot k)})$-time algorithm.
Otherwise, we can find a set of vertices which has a nonempty intersection with every solution.
Following the same approach albeit with a more refined analysis, we develop in \Cref{sec:d} an algorithm for \BFVD{} running in $\bigO^*(2^{\bigO(k^2 + \fvn \cdot k)})$ time.
That actually implies that \BFVD{} is fixed-parameter tractable for $\fvn$ when $i \ge 2$ since an instance with $k \ge \fvn$ is a \yes-instance.
In contrast, we show in \Cref{sec:hardness} that \BFVD{} is \Wone-hard for every $i \in \NN$ when parameterized by treedepth.
To the best of our knowledge, \BFVD{} is the first problem shown to be FPT for the feedback vertex number but \Wone-hard for the treedepth.
Incidentally, there are several problems that behave in the opposite way, \ie{} are FPT for the treedepth but \Wone-hard for the feedback vertex number such as \textsc{Mixed Chinese Postman} \citep{gutin2016mixed}, \textsc{Geodetic Set} \citep{Kellerhals2020}, and \textsc{Length-Bounded Cut} \citep{bentert2022length,dvovrak2018parameterized}.
Finally, we show in \Cref{ssec:bdd_fen} that \BFVD{} admits a polynomial kernel for the feedback edge number, strengthening the fixed-parameter tractability result of \citet{Betzler2012}.

\section{Preliminaries}
\label{sec:prelim}

Let~$\NN$ be the set of positive integers.
For~$n \in \NN$, let~$\oneto{n} \coloneqq \{1, 2, \dots, n\}$.

\paragraph{Graphs.}
For standard graph terminology, we refer to \citet{Diestel2017}.
For a graph $G$, let $V(G)$ denote its vertex and let $E(G)$ denote its edge set.
Let $X \subseteq V(G)$ be a vertex set. 
We denote by $G[X]$ the subgraph induced by $X$ and let $G - X$ be $G[V(G) \setminus X]$.
For an edge set $F \subseteq E(G)$, we denote by $G \setminus F$ the graph $(V(G), E \setminus F)$.
Let $N_G(X) = \{ y \mid x \in X, xy \in E(G) \} \setminus X$ and $N_G[X] = N_G(X) \cup X$.
We drop the subscript $\cdot_G$ when not ambiguous.
For simpler notation, we sometimes use $x$ for $\{ x \}$.

For fixed $i \le j \in \mathbb{N}$, we say that a pair $(S, T)$ of disjoint vertex sets is a biclique if $|S| = i$, $|T| = j$, and $st \in E(G)$ for every $s \in S$ and $t \in T$.
We refer to $S$ as the \emph{smaller side} and $T$ as the \emph{larger side}.
If $|\bigcap_{s \in S} N(s)| \ge j$, then let $(S, \cdot)$ denote an arbitrary biclique $(S, T)$ with $|T| = j$.
Let $\mathcal{S}_G$ be the collection of smaller sides of all bicliques $K_{i,j}$ of $G$ and let $\sside(G) = |\bigcup_{S \in \mathcal{S}_G} S|$.
Whenever $i$ and $j$ are clear from context, we allow ourselves to just call a $K_{i,j}$ just \emph{biclique}.

\paragraph{Graph parameters.}
The \emph{vertex cover number} $\vc(G)$ is the size of a smallest set~$V' \subseteq V(G)$ such that~$G - V'$ is edgeless.
A set $F \subseteq E(G)$ is a \emph{feedback edge set} if $G \setminus F$ is a forest.
The \emph{feedback edge number} $\fen(G)$ is the size of a smallest such set.
A set $D \subseteq V(G)$ is a \emph{feedback vertex set} if $G-D$ is a forest.
The \emph{feedback vertex number} $\fvn(G)$ is the size of a smallest such set.

For a graph $G$, a \emph{tree decomposition} is a pair $(T, B)$, where $T$ is a tree and $B \colon V(T) \to 2^{V(G)}$ such that
(i) for each edge~$uv \in E(G)$ there exists $x \in V(T)$ with $u, v \in B(x)$, and
(ii) for each $u \in V(G)$ the set of nodes $x \in V(T)$ with $v \in B(x)$ forms a nonempty, connected subtree in $T$.
The \emph{width} of $(T, B)$ is $\max_{x \in V(T)} (\abs{B(x)-1})$.
The \emph{treewidth} $\tw(G)$ of $G$ is the minimum width of all tree decompositions of~$G$.
The \emph{treedepth} of a connected graph $G$ is defined as follows~\citep{Nesetril2006}.
Let $T$ be a rooted tree with vertex set~$V(G)$, such that if $uv \in E(G)$, then $u$ is either an ancestor or a descendant of $v$ in $T$, i.e., the path from $u$ to $v$ in $T$ does not contain the root as an inner vertex.
We say that $G$ is \emph{embedded} in $T$.
The \emph{depth} of $T$ is the number of vertices in a longest path in $T$ from the root to a leaf.
The \emph{treedepth} $\td(G)$ of $G$ is the minimum $t$ such that there is a rooted tree of depth $t$ in which $G$ is embedded.

See \Cref{fig:parameters} for the relationship between parameters.
Throughout this paper, for any of the parameters $\mathrm{x}(G)$ introduced above, we allow ourselves to simply write $\mathrm{x}$ if the graph $G$ is clear from context.

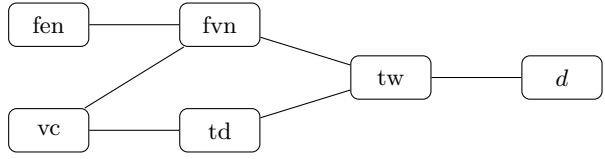
\begin{figure}
	\centering
	\begin{tikzpicture}[yscale=.7, xscale=1.5]
		\tikzset{
			param/.style={draw, fill=white, rectangle, rounded corners=3, font=\small, minimum width=3em, minimum height=3.8ex},
		}

		\node[param] at (0, 0) (vc) {$\vc$};
		\node[param] at (0, 2) (fen) {$\fen$};
		\node[param] at (1.5, 0) (td) {$\td$};
		\node[param] at (1.5, 2) (fvn) {$\fvn$};
		\node[param] at (3, 1) (tw) {$\tw$};
		\node[param] at (4.5, 1) (d) {$d$};

		\draw (fen) -- (fvn) -- (tw) -- (d);
		\draw (vc) -- (td) -- (tw);
		\draw (vc) -- (fvn);
	\end{tikzpicture}
	\caption{A Hasse diagram of parameters we study in this paper.
	An edge from $x$ (left) to $y$ (right) indicates that $x(G) \ge y(G) - 1$ for every graph $G$.
	}
	\label{fig:parameters}
\end{figure}

\subparagraph{Parameterized complexity.}
A \emph{parameterized problem} is a subset $L \subseteq \Sigma^* \times \NN$ over a finite alphabet~$\Sigma$.
Let $f \colon \mathbb N \to \mathbb N$ be a computable function.
A parameterized problem $L$ is \emph{fixed-parameter tractable (in \FPT)} with respect to $k$ if $(I, k) \in L$ is decidable in $f(k) \cdot \abs{I}^{\bigO(1)}$ time.
A \emph{kernel} for this problem is an algorithm that takes the instance $(I, k)$ and outputs a second instance $(I', k')$ such that
(i) $(I, k) \in L$ if and only if $(I', k') \in L$ and
(ii) $\abs{I'}+k' \le f(k)$ for a computable function~$f$.
The \emph{size} of the kernel is $f$.
We call a kernel \emph{polynomial} if $f$ is a polynomial.
To show that a problem $L$ is (presumably) not in \FPT, one may use a \emph{parameterized reduction} from a problem that is hard for the class $\Wone$ to $L$.
A parameterized reduction from a parameterized problem $L$ to another parameterized problem $L'$ is a function that acts as follows:
For functions $f$ and $g$, given an instance $(I, k)$ of $L$, it computes in $f(k) \cdot \abs{I}^{\bigO(1)}$ time an instance~$(I', k')$ of $L'$ so that $(I, k) \in L$ if and only if $(I', k') \in L'$ and $k' \le g(k)$.
For more details on parameterized complexity, we refer to the standard monographs~\citep{bluebook,DF13,Niedermeier06}.

\paragraph{Preliminary results.}

Our algorithms will use the following simple reduction rule.

\begin{rrule}
	\label{rr:part_of_a_biclique}
	Delete a vertex or edge that is not part of any biclique $K_{i,j}$.
\end{rrule}

Note, however, that to apply \Cref{rr:part_of_a_biclique} exhaustively, one has to enumerate all bicliques efficiently.
Although it may take $\Omega(2^n)$ time to enumerate all bicliques $(S, T)$, it is known that all \emph{maximal} bicliques can be listed in $\bigO^*(2^d)$ time in $d$-degenerate graphs \citep{eppstein1994arboricity}.
In fact, we can enumerate all vertex sets that comprise the smaller sides of all bicliques $K_{i,j}$, which we denote by $\mathcal{S}_G$, in $\bigO^*(4^d)$ time as follows:
First we enumerate all maximal bicliques.
At least one side $S'$ of each maximal biclique is of size at most $d + 1$, since a $d$-degenerate graph does not contain a biclique $K_{d+1,d+1}$ (note that every vertex has degree $d + 1$ in $K_{d+1,d+1}$).
Hence, we may assume that~$i \le d$.
For each subset $S \subseteq S'$ of size $i$ (there are $\binom{|S'|}{i} \le 2^{d+1}$ many), we add $S$ to $\mathcal{S}_G$ if $S$ is not in $\mathcal{S}_G$ and $|\bigcap_{s \in S} N(s)| \ge j$.
With $\mathcal{S}_G$ at hand, we can apply \Cref{rr:part_of_a_biclique} exhaustively with polynomial-time overhead.

After applying \Cref{rr:part_of_a_biclique} exhaustively, at least one endpoint of each edge appears in the smaller side of some biclique, resulting in the following lemma.

\begin{lemma}
	\label{lemma:s-is-vc}
	If \Cref{rr:part_of_a_biclique} has been applied exhaustively, then $\bigcup_{S \in \mathcal{S}_G} S$ is a vertex cover of~$G$.
\end{lemma}

Next, we show that \BFVD{} is FPT when parameterized by the vertex cover number $\vc(G)$.

\begin{proposition}
	\label{prop:fpt-for-vc}
	\BFVD{} can be solved in $\bigO^*(2^{\bigO(\vc \cdot k)})$ time.
\end{proposition}
\begin{proof}
	For an instance $\mathcal{I} = (G, k, i, j)$ of \BFVD{}, let $X$ be a vertex cover of $G$.
	Note that~$G-X$ is biclique-free, thus if~$\vc \le k$, we return $X$ as a solution and are done.
	So suppose that~$\vc > k$ and let $V'$ be a hypothetical solution of $\mathcal I$.
	First, our algorithm guesses the subset $X' = V' \cap X$.
	For the remaining vertices $v \in V' \setminus X$, we know that $N(v) \subseteq X$, and we guess the neighborhood of each of the at most $k$ vertices in $V' \setminus X$.
	Let $\mathcal N$ be the (multi-)set of the guessed neighborhoods.
	Note that there are at most $2^{|X|}$ choices for the first guess
	and at most $2^{|X|}$ choices for each of the neighborhoods, resulting in at most $2^{\bigO(\vc \cdot k)}$ choices.
	We arbitrarily choose a distinct vertex $v \in V(G) \setminus X$ with $N(v) = Y$ for each $Y \in \mathcal{N}$ and we delete it from $G$.
	We also delete $X'$ from the graph.
	If the resulting graph has no biclique $K_{i,j}$, which can be determined in $\bigO^*(2^d) = \bigO^*(2^{\vc})$ time, then we conclude that $\mathcal{I}$ is a \yes-instance.
\end{proof}

\Cref{lemma:s-is-vc,prop:fpt-for-vc} imply that \BFVD{} is fixed-parameter tractable with respect to $\sside(G)$ on $d$-degenerate graphs.
This fact will play an important role in the algorithm presented in \Cref{sec:d,sec:fvn}.

Finally, we show that \BFVD{} is FPT when $k, i, j, d$ are part of the parameter.
Our algorithm essentially solves an instance of \textsc{Hitting Set} in which every set has size at most $i + j$.

\begin{proposition}
	\label{obs:hitting-set-case}
	\BFVD{} can be solved in~$\bigO^*(4^d \cdot (i + j)^{k})$ time.
\end{proposition}
\begin{proof}
	We solve an instance $\mathcal{I} = (G, k, i, j)$ of \BFVD{} recursively as follows:
	If there is a biclique $(S, T)$ (which can be found in $\bigO^*(4^d)$ time), then $\mathcal{I}$ is a \yes-instance if and only if $(G - v, k - 1, i, j)$ is a \yes-instance for some $v \in S \cup T$.
	The search tree has depth at most $k$ and each node has at most $i + j$ children, and thus the running time is $\bigO^*(4^d \cdot (i + j)^{k})$.
\end{proof}

We remark that \BFVD{} is unlikely to be FPT for $i + j + k$, since it is \coWone-hard for $i + j$ when $k = 0$, as mentioned in the introduction.

\section{FPT with degeneracy and solution size}
\label{sec:d}

In this section, we show that \BFVD{} can be solved in $\bigO^*(2^{\bigO(d k^2)})$ time on graphs with degeneracy $d$, extending the known fixed-parameter tractability of \BDD{} \citep{DBLP:conf/cocoa/RamanSS08}.
Essentially, our algorithm considers two cases based on the value of $\sside(G)$.
If $\sside(G)$ is sufficiently small, then we invoke the algorithm of \Cref{prop:fpt-for-vc}.
Otherwise, we aim to find a few vertices that intersect a solution.
To find such vertices, we use the following lemma of \citet{DBLP:journals/algorithmica/AlonG09}, which has been also applied to show fixed-parameter tractability of several domination problems (including \BDD{}) \citep{DBLP:journals/algorithmica/AlonG09,DBLP:conf/wg/GolovachV08,DBLP:conf/cocoa/RamanSS08}.

\begin{lemma}[\citep{DBLP:journals/algorithmica/AlonG09}]
	\label{lemma:not-many-hit-many}
	Let $X$ be a set of at least $(4d + 2)k$ vertices.
	Then there are at most $(4d + 2)k$ vertices that are adjacent to at least $|X| / k$ vertices of $X$.
\end{lemma}

Using \Cref{lemma:not-many-hit-many}, we will show that a set of $\bigO(dk)$ vertices that intersects a hypothetical solution can be found in polynomial time whenever $\mathcal{S}_G$ is sufficiently large (\Cref{lemma:branches}).
Recall that $\mathcal S_G$ is the collection of smaller sides of all bicliques $K_{i,j}$.
The proof of \Cref{lemma:branches} relies on the following lemma.

\begin{lemma}
	\label{lemma:must-hit-many}
	Let $\mathcal{I} = (G, k, i, j)$ be a \yes-instance of \BFVD{}.
	Let $\mathcal{X} \subseteq \mathcal{S}_G$ be a nonempty collection of smaller sides of bicliques and let $X = \bigcup_{X' \in \mathcal{X}} X'$.
	Suppose that $V'$ is a solution of $\mathcal{I}$ with $V' \cap X = \emptyset$.
	Then there exists a vertex $v \in V'$ which has at least $|X| / k$ neighbors in $X$.
\end{lemma}
\begin{proof}
	Assume to the contrary that every vertex $v$ in $V'$ has less than $|X| / k$ neighbors in $X$, i.e., $|X \cap N(v)| < |X| / k$.
	Then we have $X \setminus N(V') \ne \emptyset$ since
	\[
		|X \setminus N(V')| \ge |X| - \Big| \bigcup_{v \in V'} (X \cap N(v)) \Big| \ge |X| - \sum_{v \in V'} |X \cap N(v)| > 0.
	\]
	Choose any $u \in X \setminus N(V')$.
	By the definition of $X$, there exists a biclique $(S, T)$ such that $u \in S$ and $S \subseteq X$.
	Since $V' \cap X = \emptyset$, the solution $V'$ does not intersect $S$.
	It does not intersect $T$ either, since every vertex in $T$ is adjacent to $u$.
	Thus, there remains a biclique $K_{i,j}$ in $G - V'$, a contradiction.
\end{proof}

\begin{lemma}
	\label{lemma:branches}
	Let $\mathcal{I} = (G, k, i, j)$ be an \yes-instance of \BFVD{}.
	If $\sside(G) > (4d + 2)k$, then we can find in polynomial time a set $W$ of at most $(8d + 4)k + i$ vertices such that $V' \cap W \ne \emptyset$ for every solution $V'$ of~$\mathcal{I}$.
\end{lemma}
\begin{proof}
	We first find an inclusion-wise minimal subcollection $\mathcal{X} \subseteq \mathcal{S}_G$ of smaller sides of bicliques such that $|\bigcup_{X' \in \mathcal{X}} X'| \ge (4d + 2)k$.
	This can be done in polynomial time using a simple greedy algorithm.
	Let~$X = \bigcup_{X' \in \mathcal X} X'$.
	Note that~$|X| \le (4d + 2)k + i$ since otherwise the deletion of arbitrary $X' \in \mathcal X$ from $\mathcal X$ gives us another desired subcollection, contradicting the minimality of $\mathcal X$.
	We will include $X$ into~$W$.
	Thus, $V' \cap W \ne \emptyset$ holds if $V' \cap X \ne \emptyset$. 
	If $V' \cap X = \emptyset$, then by \Cref{lemma:must-hit-many}, there exists a vertex $v \in V'$ which has at least $|X| / k$ neighbors in $X$.
	Let $U$ be the set of vertices $u$ with at least $|X| / k$ neighbors in $X$.
	By \Cref{lemma:not-many-hit-many}, $|U| \le (4d + 2)k$.
	Thus, the lemma holds for $W = X \cup U$.
\end{proof}

Finally, we show that \BFVD{} is FPT with respect to $k + d$ using \Cref{lemma:branches}.

\begin{theorem}
	\label{thm:fpt-for-dk}
	\BFVD{} can be solved in time $\bigO^*(2^{\bigO(dk^2)})$.
\end{theorem}
\begin{proof}
	Given an instance $(G, k, i,j)$ of \BFVD{}, we first apply \Cref{rr:part_of_a_biclique} exhaustively.
	If $i > d$, then we obtain a trivial \yes-instance since $G$ does not have $K_{d+1,d+1}$.
	We consider two cases.
	Suppose first that $\sside(G) \le (4d + 2)k$.
	Since $\bigcup_{S \in \mathcal{S}_G} S$ is a vertex cover of $G$ by \Cref{lemma:s-is-vc}, we can solve the instance using the algorithm of \Cref{prop:fpt-for-vc} in $\bigO^*(2^{\bigO(dk^2)})$ time.
	Otherwise, we have $\sside(G) > (4d + 2)k$.
	In this case, we can find in polynomial time a set $W$ of size at most $(8d + 4)k + i \in \bigO(dk)$ vertices that intersects every solution by \Cref{lemma:branches}.
	For every vertex $w \in W$, we recursively solve the instance $(G - w, k - 1, i, j)$ until we have a trivial instance or $\sside(G) \le (4d + 2)k$.
	The running time is bounded by $\bigO^*(2^{\bigO(dk^2)})$.
\end{proof}

\section{FPT with feedback vertex number}
\label{sec:fvn}

In the previous section, we have seen an $\bigO^*(2^{\bigO(dk^2)})$-time algorithm for \BFVD{}.
We present an algorithm for \BFVD{} running in $\bigO^*(2^{\bigO(k^2 + \fvn \cdot k)})$ time in this section:

\begin{theorem}
	\label{thm:fpt-fvn}
	For $i \ge 2$, \BFVD{} can be solved in $\bigO^*(2^{\bigO(k^2 + \fvn \cdot k)})$ time.
\end{theorem}

For $i \ge 2$, any instance $(G, k, i,j)$ with $k \ge \fvn$ is a \yes-instance, since a forest does not contain any biclique $K_{i,j}$.
Thus, we have the following corollary:

\begin{corollary}
	For $i \ge 2$, \BFVD{} can be solved in $\bigO^*(2^{\bigO(\fvn^2)})$ time.
\end{corollary}

We remark that, as the degeneracy of a graph is at most~$\fvn+1$, this is faster than the running time of the algorithm derived analogously from \Cref{thm:fpt-for-dk}, which is $\bigO^*(2^{\bigO(\fvn^3)})$.

\cref{alg:fpt-fvn} provides an overview of the algorithm that shows \cref{thm:fpt-fvn}.
In a nutshell, we identify several cases that are efficiently solvable.
If none of the cases apply, then $\sside(G) \in \bigO(k + \fvn)$ holds, and we can use the algorithm of \Cref{prop:fpt-for-vc}.
Let $V'$ denote a hypothetical solution and let $D$ be a minimum feedback vertex set.
We first guess the intersection $D' = V' \cap D$, which we delete from the graph.
Let $R \subseteq V(G) \setminus D$ be the set of vertices whose closed neighborhood contains at least three vertices in $\bigcup_{S \in \mathcal{S}_G} S$.
As we show later in \Cref{lem:large-r}, we can immediately conclude that we have a \no-instance if $|R| > 3k$ (\Cref{line:fpt-fvn:r-no}).
Again, we guess the intersection $R' = V' \cap R$ to be deleted from the graph.
If more than $2k$ vertices in the forest $F = G-D$ remain in $\bigcup_{S \in \mathcal{S}_G} S$, then we can conclude that the instance has no solution (\Cref{line:fpt-fvn:q-no}), as we show in \Cref{lem:large-q}.

\begin{algorithm}[t!]
	\caption{The algorithm for \cref{thm:fpt-fvn}. We assume that \Cref{rr:part_of_a_biclique} is exhaustively applied throughout.}
	\label{alg:fpt-fvn}
	\small
	\SetAlgoLined
	\SetKwInOut{Input}{Input}
	\SetKwInOut{Output}{Task}
	\Input{A graph $G$, integers $i \le j$, $k$, and a feedback vertex set $D$.}

	\lIf{$j \le \fvn + 1$}{\KwRet the result of the algorithm of \cref{obs:hitting-set-case}.} \label{line:fpt-fvn:hitting-set} 
	\textbf{guess} $D' \subseteq D$. Remove~$D'$ from $G$ and~$D$, set $k \gets k-\abs{D'}$.\; \label{line:fpt-fvn:guess-d}
	$F \gets G-D$.
	Root $F$ arbitrarily.\;
	$R \gets \{v \in V(F) \mid \abs{N_F[v] \cap \bigcup_{S \in \mathcal{S}_G} S} \ge 3 \}$.\;
	\lIf{$\abs{R} > 3k$}{\KwRet \no. (\Cref{lem:large-r})} \label{line:fpt-fvn:r-no}
	\textbf{guess} $R' \subseteq R$, $\abs{R'} \le k$. Remove~$R'$ from $G$ and $F$, set $k \gets k-\abs{R'}$.\; \label{line:fptt-fvn:guess-r}
	\lIf{$\abs{V(F) \cap \bigcup_{S \in \mathcal{S}_G} S} > 2k$}{\KwRet \no. (\Cref{lem:large-q})} \label{line:fpt-fvn:q-no}
	\KwRet the result of the algorithm of \Cref{prop:fpt-for-vc}.\; \label{line:final}
\end{algorithm}

The following observation that at most one vertex of $F$ appears in the smaller side of a biclique becomes crucial to establish the correctness of \Cref{alg:fpt-fvn}.

\begin{observation}
	\label{obs:q-center}
	If~$j > \fvn + 1$, then the smaller side~$S$ of every biclique~$(S, T)$ contains at most one vertex of~$V(F)$.
\end{observation}
\begin{proof}
	If~$j > \fvn + 1$, then for each biclique $(S, T)$, there are two vertices $u, v \in T \cap V(F)$.
	If there are two vertices in $S \cap V(F)$, then they induce a cycle with $u$ and $v$ in the forest, a contradiction.
\end{proof}

The following lemma shows that \Cref{line:fpt-fvn:r-no} is correct.
Since we delete the intersection $D' = V' \cap D$ from the graph in \Cref{line:fpt-fvn:guess-d}, we may assume that $V' \subseteq V(F)$.

\begin{lemma}
	\label{lem:large-r}
	If $\abs{R} > 3k$ in \cref{line:fpt-fvn:r-no}, then every set~$V' \subseteq V(F)$ that intersects every~$K_{i,j}$ contains more than $k$ elements.
\end{lemma}
\begin{proof}
	Partition~$R$ into three sets~$R_0, R_1, R_2$ such that~$v$ is in~$R_\delta$ if the distance from~$v$ to the root of the same component is $\delta$ modulo $3$.
	At least one of the partitions, say $R_0$, contains more than $k$ elements.
	By the definition of $R$, we have $|N_F[v] \cap \bigcup_{S \in \mathcal{S}_G} S| \ge 3$ for every $v \in R$.
	Note that $N_F[v]$ consists of $v$ itself, its parent (if it exists), and its child(ren).
	It follows that $v$ has a child $q_v \in \bigcup_{S \in \mathcal{S}_G} S$.
	Let $(S_v, T_v)$ be an arbitrary biclique $K_{i,j}$ with $q_v \in S_v$ and let $U_v = S_v \cup T_v$.
	By \Cref{obs:q-center}, we have $S_v \cap V(F) = \{ q_v \}$.
	Thus, $T_v \cap V(F) \subseteq N_F(q_v)$.
	As $q_v \in R_1$, we have $N_F(q_v) \cap R_0 = \{v\}$, and all remaining neighbors of $q_v$ are in $R_2$.
	Now pick $v' \in R_0$ with $v' \ne v$ and define $S_{v'}$, $T_{v'}$, $U_{v'}$ and $q_{v'}$ analogously.
	Then $N_F[q_v] \cap N_F[q_{v'}] = \emptyset$ as $v \ne v'$ and $q_v \ne q_{v'}$;
	the remaining vertices in the neighborhoods are children of either $q_v$ or $q_{v'}$ and thus cannot be equal either.
	Consequently $U_v \cap V(F)$ and $U_{v'} \cap V(F)$ are disjoint.
	Hence, a set $V' \subseteq V(F)$ intersecting every biclique contains at least one vertex of $U_v \cap V(F)$ for every vertex in $v \in R_0$.
	Thus, $\abs{V'} \ge \abs{R_0} > k$.
\end{proof}

Next, we show that \Cref{line:fpt-fvn:q-no} is correct.
In \Cref{line:fptt-fvn:guess-r}, we delete the vertices of $R$ included in the hypothetical solution $V'$.
We thus may assume that $V'$ is disjoint from $R$.

\begin{lemma}
	\label{lem:large-q}
	If $|V(F) \cap \bigcup_{S \in \mathcal{S}_G} S| > 2k$ in \cref{line:fpt-fvn:q-no}, then every set~$V' \subseteq V(F) \setminus R$ that intersects every $K_{i,j}$ contains more than $k$ elements.
\end{lemma}
\begin{proof}
	Suppose that there exists a set $V' \subseteq V(F) \setminus R$ that intersects every $K_{i,j}$ of size at most $k$.
	By the definition of $R$, every vertex $v' \in V'$ has $|N_F[v'] \cap \bigcup_{S \in \mathcal{S}_G} S| \le 2$.
	As $|V(F) \cap \bigcup_{S \in \mathcal{S}_G} S| > 2k \ge 2 |V'|$, we have
	\[
		\Big\lvert\Big(V(F) \cap \bigcup_{S \in \mathcal{S}_G} S\Big) \setminus N[V']\Big\rvert \ge \Big\lvert V(F) \cap \bigcup_{S \in \mathcal{S}_G} S\Big\rvert - \sum_{v' \in V'} \big\lvert N_F[v'] \cap \bigcup_{S \in \mathcal{S}_G} S\big\rvert > 0.
	\]
	This implies the existence of a vertex $v \in (V(F) \cap \bigcup_{S \in \mathcal{S}_G} S) \setminus N[V']$.
	Let $(S_v, T_v)$ be an arbitrary biclique with $v \in S_v$ and let $U_v = S_v \cup T_v$.
	By \Cref{obs:q-center}, we have $U_v \cap V(F) \subseteq N_F[v]$.
	Since $v \notin N[V']$, $V' \subseteq V(F)$ does not intersect $U_v$, a contradiction.
\end{proof}

Finally, we analyze the running time of \Cref{alg:fpt-fvn}.

\begin{lemma}
	\label{lem:fpt-fvn-time}
	\cref{alg:fpt-fvn} runs in time $\bigO^*(2^{\bigO(k^2 + \fvn \cdot k)})$.
\end{lemma}
\begin{proof}
	In \Cref{alg:fpt-fvn}, we guess $D'$ and $R'$.
	There are $2^{\bigO(\fvn)}$ choices for~$D'$ and $2^{3k}$ choices for $R'$, amounting to $2^{\bigO(k + \fvn)}$ choices.
	Since we have $\sside(G) \le |V(F) \cap \bigcup_{S \in \mathcal{S}_G} S| + |D \cap \bigcup_{S \in \mathcal{S}_G} S| \le 3k + \fvn$ in \Cref{line:final} and $\bigcup_{S \in \mathcal{S}_G} S$ is a vertex cover of $G$ by \Cref{lemma:s-is-vc}, the algorithm of \Cref{prop:fpt-for-vc} runs in time $\bigO^*(2^{\bigO(k^2 + \fvn \cdot k)})$.
	We spend $\bigO^*(4^{\fvn})$ time elsewhere; hence \Cref{alg:fpt-fvn} runs in the claimed time.
\end{proof}

The correctness of \Cref{alg:fpt-fvn} follows from \Cref{obs:hitting-set-case,lem:large-q,lem:large-r,prop:fpt-for-vc} and it runs in time $\bigO^*(2^{\bigO(k^2 + \fvn \cdot k)})$ by \Cref{lem:fpt-fvn-time}.
This proves \Cref{thm:fpt-fvn}.

\section{Parameterized hardness}
\label{sec:hardness}
\citet{Ganian2021} show that \BFVD{} is \Wone-hard with respect to the tree\-depth of the input graph if $i=1$.
We show that this holds true for every fixed value of $i$.

\begin{theorem}
	For every fixed $i$, \BFVD{} is \Wone-hard when parameterized by the tree\-depth.
\end{theorem}
\begin{proof}
	We reduce from \BDD{}, which is \Wone-hard when parameterized by tree\-depth \citep{Ganian2021}.
	Given an instance $(G, k, r)$ of \BDD{}, we construct an instance $(G', k, i, j)$ as follows.
	We set~$j = n + r + 1$, where $n$ is the number of vertices in $G$.
	We will assume that $n > i$.
	For the construction of $G'$, we start with a copy of $G$.
	For every $v \in V(G)$, we introduce $i - 1$ vertices $S_v = \{ s_v^1, \dots, s_v^{i - 1} \}$ and $n$ vertices $T_v = \{ t_v^1, \dots, t_v^{n} \}$ and add edges such that $S_v' = \{ v \} \cup S_v$ and $T_v$ form a biclique $K_{i,n}$.
	Moreover, for every edge $uv \in E(G)$, we add an edge between $u$ and $s_v$ for every $s_v \in S_v$.
	
	Suppose that $(G, k, r)$ has a solution $V'$.
	We claim that $V'$ is also a solution of $(G', k, i, j)$.
	Suppose to the contrary that $G - V'$ has a biclique $K_{i, j}$.
	Then, its smaller side is $S_v'$ for some vertex $v \in V(G)$ and its larger side is a subset of $(N_G(v) \setminus V') \cup T_v$.
	To see why, observe that $V(G) \cup \bigcup_{v \in V(G)} S_v$ constitutes the set of vertices of degree at least $n > i$ and that two vertices have at least $n$ common neighbors in $G'$ if and only if they belong to the same $S_v'$.
	In particular, it holds that $v \notin V'$. 
	Since every vertex in $V \setminus V'$ has degree at most $r$ in $G - V'$, we have $|N_G(v) \setminus V'| \le r$ and thus $|(N_G(v) \setminus V') \cup T_v| \le n + r$, a contradiction.

	Conversely, suppose that $V'$ is a solution of $(G', k, i, j)$.
	We claim that the set $V'' = \{ v \in V(G) \mid (S_v' \cup T_v) \cap V' \ne \emptyset \}$ is a solution of $(G, k, r)$.
	Note that $|V''| \le |V'| \le k$.
	Suppose that there exists a vertex $v \in V(G) \setminus V'$ of degree greater than $r$ in $G - V''$.
	Then, $S_v'$ is of size $i$ and it has at least $|(N_G(v) \setminus V') \cup T_v| \ge j$ common neighbors.
	We thus conclude that $G - V''$ is a solution.

	Finally, we show that $\td(G') \le i \cdot \td(G) + 1$ by providing a rooted tree $T'$ of depth $\td(G')$ in which $G'$ is embedded.
	The tree is based on a rooted tree $T$ of depth $\td(G)$ in which $G$ is embedded.
	We replace every vertex~$v \in V(T)$ with a path consisting of the vertices in~$S'_v$ and attach the children of $v$ in $T$ to the lowermost (furthest from the root) vertex in~$S'_v$.
	Then we add each~$t_v \in T_v$ as a leaf to the lowermost vertex in~$S'_v$.
	Note that any ancestor $u \in V(G)$ of $v$ is now ancestor of all vertices in~$S'_v$, and each vertex in~$S'_v$ is ancestor of each vertex in~$T_v$;
	thus $G'$ is embedded in $T'$.
	As we replace every vertex with a path of length $i$, and attach at most one child at the bottom of the path, the depth of $T'$ is at most $i \cdot \td(G) + 1$.
\end{proof}

\section{Polynomial kernel with respect to feedback edge number}
\label{ssec:bdd_fen}
In this section, we show that \BFVD{} admits a polynomial kernel when parameterized by the feedback edge number $\fen$.

\begin{theorem}
	\label{thm:bdd-fen-kernel}
	\BFVD{} admits a kernel of size $\bigO(\fen^2)$ for $i = 1$ and $\bigO(\fen)$ for $i \ge 2$.
\end{theorem}

We start with the case $i = 1$. 
Then \BFVD{} is equivalent to \textsc{Bounded-Degree Deletion (BDD)} as there is a trivial parameter-preserving reduction (set $r = j - 1$, where $r$ is the degree bound of BDD and $j$ is the size of one of the biclique sides).
It is known that \BDD{} is fixed-parameter tractable for $\fen$ \citep{Betzler2012}.
We strengthen their result proving the existence of a polynomial kernel.

To develop a kernelization algorithm, we will work with the following generalization of \BDD{}.

\newcommand{\WBDD}{\textsc{Weighted Bounded-Degree Deletion}}
\newcommand{\wbdd}{\prob{\textsc{WBDD}}}
\problemdef{\WBDD{} (\wbdd{})}
{An undirected graph~$G$, two integers~$k, r \in \NN$, and weights~$w \in \NN^{V(G)}$.}
{Does there exist a subset~$V'\subseteq V(G)$ with~$\abs{V'}\le k$ such that each vertex~$v\in V(G)\setminus V'$ has degree at most~$r-w_v$ in~$G-V'$?}

Herein, by $w_v$ we denote the weight of a vertex $v$.
Note that \BDD{} is a special case of \wbdd{}, where $w_v = 0$ for each $v \in V$.

We use the weights in the following manner:
Suppose that for an instance of \wbdd{}, we identify a vertex $v$ which can be ``avoided'', that is, there is a solution $V'$ with $v \notin V'$.
Then we can simplify the instance as follows:
delete $v$ and increase the weight of every neighbor of $v$ by one.

To show that \BDD{} admits a kernel of size $\bigO(\fen^2)$, we first show that \wbdd{} has a kernel of size $\bigO(\fen)$ for constant $r$.
We then show how, given a \wbdd{} instance of size $\bigO(\fen)$, we can transform it into a \BDD{} instance of size $\bigO(\fen^2)$.

\paragraph{Linear kernel for WBDD.}
As a first step to obtain a linear kernel for \wbdd{}, we apply reduction rules based on $\deg(v)$ and $w_v$.
We first observe that our problem treats a vertex~$v$ the same whenever~$\deg(v)+w_v \le r$.
Hence, we may set the weight $w_v$ of such vertices to~$r-\deg(v)$.

\begin{rrule}
	\label{rr:wbdd_low_weight_rule}
	If $\deg(v) + w_v < r$, then increase~$w_v$ by one.
\end{rrule}

Next, if a weight of a vertex is too high, then it must be in any solution.
\begin{rrule}
	\label{rr:wbdd_high_weight_rule}
	If~$w_v > r$, then delete~$v$ and decrease~$k$ by one. %
\end{rrule}

After applying these two reduction rules, we have $r - \deg(v) \le w_v \le r$ for every vertex $v$.
In particular, we have $w_v = r$ for every isolated vertex $v$, which can be deleted.

\begin{rrule}
	\label{rr:wbdd_deg_0}
	Let~$v\in V$ be an isolated vertex of~$G$ with~$w_v = r$.
	Then, delete~$v$.
\end{rrule}
For a degree-one vertex $v$, we have~$w_v = r-1$ or~$w_v = r$.
In either case, it does not make sense to take~$v$ into the solution, as deleting its neighbor affects the degrees of at least as many vertices.

\begin{rrule}
	\label{rr:wbdd_deg_1}
	Let~$v\in V$ be a vertex of~$G$ with~$N(v) = \{u\}$.
	If $w_v = r - 1$, then delete $v$ and increase $w_u$ by one.
	If~$w_v = r$, then delete both~$v$ and~$u$, and decrease~$k$ by one. %
\end{rrule}

\begin{lemma} %
	\label{lm:wbdd_deg_rules}
	\cref{rr:wbdd_low_weight_rule,rr:wbdd_high_weight_rule,rr:wbdd_deg_0,rr:wbdd_deg_1} are correct and can be applied exhaustively in~$\bigO(n+m)$ time.
\end{lemma}
\begin{proof}
	Suppose~$V' \subseteq V(G)$ is a solution for an instance of \wbdd{} and consider a vertex~$v \in V(G)$.
	Suppose that~$\deg(v) + w_v \le r$.
	Then, $V'$ remains a solution if we replace~$w_v$ with a weight~$w'_v$, with $0 \le w'_v \le r - \deg(v)$.
	Hence, \cref{rr:wbdd_low_weight_rule} is correct.
	Suppose next that~$w_v > r$.
	Then any solution~$V'$ must contain~$v$.
	Hence, $G-\{v\}$ contains a solution of size~$k-1$ if and only if~$G$ contains a solution of size~$k$, and \cref{rr:wbdd_high_weight_rule} is correct.
	Suppose now that~$N(v) = \{u\}$.
	If~$w_v = r$, then either~$u \in V'$ or~$v \in V'$.
	As the choice of~$u \in V'$ decreases the degree of at least as many vertices as~$v \in V'$, we may always pick~$u$ into~$V'$.
	Suppose that $w_v = r-1$.
	If $v \in V'$, then $(V' \setminus \{v\}) \cup \{u\}$ is also a valid solution by the same argument;
    Otherwise, $V'$ remains a solution.
	Conversely, any solution for the resulting instance is a solution for the original instance because we increase the weight $w_u$ by one; thus \cref{rr:wbdd_deg_1} is correct.
	The correctness of \cref{rr:wbdd_deg_0} is obvious.

	To apply the reduction rules exhaustively, we have to test for each vertex if one of the reduction rules can be applied
	and repeat this test for all vertices where at least one neighbor got deleted by applying one of the rules.
	This can be done in~$\bigO(n+m)$ time by checking the weight and the degree of each vertex
	on a list of vertices that still need to be tested.
	If a vertex gets deleted, then all neighbors that are not on the list already need to be added again.
	The actual exhaustive application of the rules requires~$\bigO(n+m)$ time,
	since deleting all vertices while still maintaining a correct graph representation
	is possible in this time.
	Combining the steps leads to~$\bigO(n+m)$ time in total.
\end{proof}

We will henceforth assume that \cref{rr:wbdd_low_weight_rule,rr:wbdd_high_weight_rule,rr:wbdd_deg_0,rr:wbdd_deg_1} have been exhaustively applied.

To obtain a linear kernel for \wbdd{}, we use the following folklore result.
We call a path \emph{maximal} if both of its endpoints have degree at least three and all inner vertices have degree exactly two.
We call a cycle \emph{maximal} if at most one of its vertices has degree at least three.

\begin{lemma}[See e.g., \citep{Epstein2015,Kellerhals2020}]
	\label{lm:v3}
	Let~$G$ be a graph in which each vertex has degree at least two.
	Then, the number of vertices of degree at least three is at most $2 \fen - 2$.
	Moreover, the number of maximal paths and cycles in $G$ is at most $3 \fen - 3$.
\end{lemma}

By \Cref{lm:v3}, the number of vertices of degree at least three is at most $2 \fen - 2$.
It remains to bound the length of maximal paths and cycles in which each vertex has degree two.
We introduce further notation.
For a (vertex-) weighted graph $(G, w)$, let $\opt(G, w)$ denote the minimum integer $k$ such that $(G, k, r, w)$ is a \yes-instance of \wbdd{}.
If $G$ is a path, then $\opt(G, w)$ is linear-time computable by a trivial adaptation of \cref{rr:wbdd_low_weight_rule,rr:wbdd_high_weight_rule,rr:wbdd_deg_0,rr:wbdd_deg_1}.

Our algorithm works as follows:
If the graph contains a sufficiently long path $P$ of degree-two vertices, then we replace it with another weighted path $P'$ of shorter length.
The replacement path $P'$ should behave analogously to $P$ in the context of \wbdd.
Our key finding is that the \emph{characteristic matrix} of weighted degree-two paths determines the behavior of weighted paths.
Intuitively, the characteristic matrix captures the increase in the optimal solution size when a subset of the four outermost vertices (i.e., the endpoints and each of their neighbors) is included.

\begin{definition}
	\label{def:charm}
	For an integer $\ell \ge 5$,  let $\mathcal{P}_\ell$ denote the collection of weighted paths $(P, w)$ on $\ell$ vertices $v_1, \dots, v_\ell$ such that $w_{v_1} = w_{v_\ell} = r - 1$ and $w_{v_i} \in \{ r - 2, r - 1, r \}$ for each $i \in [2, \ell - 1]$.
	For a weighted path $(P, w) \in \mathcal{P}_{\ell}$, the \emph{characteristic matrix} of $(P, w)$ is a $3 \times 3$ matrix $M(P, w)$ such that $M(P, w)_{x, y} = s_{x, y} - \opt(P, w)$, where $s_{x, y}$ is the minimum size of a vertex set $S$ such that 
	\begin{enumerate}[(i)]
		\item $v_1 \in S$ if $x = 1$, $v_2 \in S$ and~$v_1 \not\in S$ if $x = 2$ and $v_1, v_2 \notin S$ if $x = 3$,
		\item $v_\ell \in S$ if $y = 1$, $v_{\ell - 1} \in S$ and~$v_{\ell} \not\in S$ if $y = 2$ and $v_{\ell}, v_{\ell - 1} \notin S$ if $y = 3$, and
		\item for every vertex $v$ in $P - S$, $\deg_{P - S}(v) + w_v \le r$.
	\end{enumerate}
	Here, we assume that $s_{x, y} = \infty$ if there exists no set fulfilling (i), (ii), and (iii). 
\end{definition}

Note that, given a weighted path $(P, w) \in \mathcal{P}_{\ell}$, we can compute its characteristic matrix in linear time by adapting \cref{rr:wbdd_high_weight_rule,rr:wbdd_deg_0,rr:wbdd_deg_1}.
We verified the following lemma using a computer program which enumerates the characteristic matrices of all weighted paths in $\mathcal P_\ell$, $\ell \le 7$.\footnote{The source code is made available at \url{https://git.tu-berlin.de/akt-public/bfvd-kernel}}
We remark that there are 11 distinct characteristic matrices arising from weighted paths on seven vertices.

\begin{lemma}
	\label{lemma:shorter_path_base}
	For every weighted path $(P, w) \in \mathcal{P}_{7}$, there exists a weighted path $(P', w') \in \mathcal{P}_{6}$ such that $M(P, w) = M(P', w')$.
\end{lemma}

Observe that, given a weighted path $(P, w) \in \mathcal P_{7}$, we can compute a shorter weighted path $(P', w')$ such that $M(P, w) = M(P', w')$ in $\bigO(1)$ time.
With this at hand, we can show that every maximal path can be replaced by a path on at most $6$ vertices.
For this, we need some additional notation.
Let~$(G, w)$ be a graph and let~$P = (v_1, \dots, v_\ell)$ be a path in~$G$.
We call~$P - \{v_1, v_\ell\}$ the \emph{inner path} of $P$ and~$V(P - \{v_1, v_\ell\})$ its \emph{inner vertices}.
Let~$w^* \subseteq \NN^{V(P)}$ be the weight vector obtained from~$w$ by restricting it to~$V(P)$ and replacing the weights of~$v_1$ and~$v_\ell$ with $r-1$,
that is, $w^*_{v_1} = w^*_{v_\ell} = r-1$ and~$w^*_{v_i} = w_{v_i}$ for each inner vertex~$v_i$.

\begin{rrule}
	\label{rr:wbdd_path}
	Let $P = (v_1, \dots, v_{7})$ be a (not necessarily maximal) path whose inner vertices have all degree two.
	Let $(P', w') \in \mathcal{P}_{6}$ be a weighted path such that $M(P', w') = M(P, w^*)$.
	Then replace the inner path of $P$ with the inner path of $P'$ and decrease $k$ by $\opt(P, w^*) - \opt(P', w')$.
\end{rrule}

\begin{lemma}
	\label{lemma:wbdd_path}
	\Cref{rr:wbdd_path} is correct.
\end{lemma}
\begin{proof}
	Let~$\mathcal I = (G, k, r, w)$ be the instance of \wbdd{}, which contains a path $P = (v_1, \dots, v_{7})$ as described in \cref{rr:wbdd_path},
	and let~$\mathcal I'$ be the instance of \wbdd{} obtained from executing \cref{rr:wbdd_path}.
	First, note that the existence of a weighted path $(P', w') \in \mathcal{P}_{\ell}$ in \Cref{rr:wbdd_path} is guaranteed by \Cref{lemma:shorter_path_base}. 
	Suppose that $\mathcal{I}$ has a solution $S$.
	Let
	\[
		x = \begin{cases}
		1 & \text{if } v_1 \in S,\\
		2 & \text{if } v_1 \notin S, v_2 \in S,\\
		3 & \text{if } v_1, v_2 \not\in S,
		\end{cases}
		\quad \text{and} \quad
		y = \begin{cases}
			1 & \text{if } v_7 \in S,\\
			2 & \text{if } v_6 \in S, v_7 \not\in S,\\
			3 & \text{if } v_6, v_7 \not\in S.
		\end{cases}
	\]
	Note that $|S \cap \{ v_1, \dots, v_{7} \}| \ge \opt(P, w^*) + M(P, w^*)_{x, y}$ by \Cref{def:charm}.
	In particular, it holds that $M(P, w^*)_{x, y} \ne \infty$.
	By the assumption that $M(P, w^*)= M(P', w')$, there exists a subset $Q$ of vertices in $P' = (v_1', \dots, v'_{6})$ such that
	\begin{enumerate}[(i)]
		\item $v_1' \in Q$ if $x = 1$, $v_2' \in Q$ and~$v_1' \not\in Q$ if $x = 2$, and $v_1', v_2' \notin Q$ if $x = 3$,
		\item $v_6' \in Q$ if $y = 1$, $v_{5}' \in Q$ and~$v_{6}' \not\in Q$ if $y = 2$, and $v_{5}', v_{6}' \notin Q$ if $y = 3$, 
		\item for every vertex $v$ in $P' - Q$, $\deg_{P' - Q}(v) + w_v \le r$, and
		\item $|Q| = \opt(P', w') + M(P, w^*)_{x, y}$.
	\end{enumerate}

	We claim that $S' = (S \setminus \{ v_2, \dots, v_{6} \}) \cup (Q \setminus \{ v_1', v_{6}' \})$ is a solution of $\mathcal{I}'$.
	Since $|S \cap \{ v_1, v_{7} \}| = |Q \cap \{ v_1', v_{6}' \}|$, we have 
	\begin{align*}
		|S'|
		&= |S| - |S \cap \{ v_2, \dots, v_{6} \}| + |Q \setminus \{ v_1', v_{6}' \}| \\
		&= |S| - \big(|S \cap \{ v_2, \dots, v_{6} \}| + |S \cap \{ v_1, v_{7} \}|\big)+ \big(|Q \setminus \{ v_1', v_{6}' \}| + |Q \cap \{ v_1', v_{6}' \}|\big) \\
		&= |S| - |S \cap \{ v_1, \dots, v_{7} \}| + |Q| \\
		&\le k - \big(\opt(P, w^*) + M(P, w^*)_{x, y}\big) + \big(\opt(P', w') + M(P, w^*)_{x, y}\big) \\
		&= k - \big(\opt(P, w^*) - \opt(P', w')\big).
	\end{align*}
	To verify that $S'$ is a solution, it suffices to show that
	(1) $v_1 \in S'$ or $S'$ contains at least $\deg(v_1) + w_{v_1} - r$ neighbors of $v'_1$ in $G'$ and
	(2) $v_{7} \in S'$ or $S'$ contains at least $\deg(v_{7}) + w_{v_{7}} - r$ neighbors of $v_{7}$ in $G'$.
	(Note that~$G'$ still contains the ``original'' endpoints $v_1$ and~$v_7$ of~$P$.)
	We only prove (1), since (2) can be shown analogously.
	If $x = 1$, then (1) clearly holds as $v_1 \in S$.
	If $x = 2$, then $v_1 \notin S$ and $v_2 \in S$;
	thus $S \setminus \{v_2, \dots, v_6\}$ contains at least $\deg(v_1) + w_{v_1} - r - 1$ neighbors of $v_1$,
	and as $S \setminus \{v_2, \dots, v_6\} \subseteq S'$ and~$v'_2 \in S'$, the set~$S'$ contains at least $\deg(v_1) + w_{v_1} - r$ neighbors of $v_1$.
	If $x = 3$, then $v_1, v_2 \notin S$;
	thus $S \setminus \{v_2, \dots, v_6\}$ contains at least~$\deg(v_1) + w_{v_1} - r$ neighbors of~$v_1$,
	and as $S \setminus \{v_2, \dots, v_6\} \subseteq S'$, the same holds for~$S'$.

	The other direction can be shown analogously because the proof of the forward direction does not rely on the fact that $P'$ is shorter than $P$.
\end{proof}

Note that we can apply \Cref{rr:wbdd_path} exhaustively in linear time since we have at most $|V(G)|$ applications of \Cref{rr:wbdd_path}, each of which take $\bigO(1)$ time to compute.

\begin{proposition} %
	\label{prop:wbdd_kernel}
	For constant $r$, \wbdd{} has a kernel of size $\bigO(\fen)$.
\end{proposition}
\begin{proof}
	We claim that after we apply \cref{rr:wbdd_high_weight_rule,rr:wbdd_deg_0,rr:wbdd_deg_1,rr:wbdd_path} exhaustively, we have an instance where the graph is of size $\bigO(\fen)$.
	Note that \Cref{lm:wbdd_deg_rules,lemma:wbdd_path} establish the correctness of our rules.
	Moreover, we can apply these rule in linear time.
	Since \cref{rr:wbdd_high_weight_rule,rr:wbdd_deg_0,rr:wbdd_deg_1} delete all vertices of degree at most one, we have at most $2\fen - 2$ vertices of degree at least three by \Cref{lm:v3}.
	Moreover, we have at most $3 \fen - 3$ maximal paths and cycles whose internal vertices have degree two.
	By \Cref{rr:wbdd_path}, such a path or cycle is of length at most eleven.
	Since each degree-two vertex and each edge is contained in such a maximal path or cycle, the graph is of size $\bigO(\fen)$.
	Finally, note that we need $\bigO(1)$ bits to encode each vertex weight.
\end{proof}

\paragraph{Removing weights.}
Towards showing that \BDD{} has a kernel of size $\bigO(\fen^2)$, we use the following reduction rule to ensure that the weight of every vertex is at most $\bigO(\fen)$.

\begin{rrule}
	\label{rr:weight}
	If $w_v > 0$ for every vertex $v$, then decrease each $w_v$ by one and decrease $r$ by one.
\end{rrule}

\begin{lemma} %
	\label{lemma:rr:weight}
	\Cref{rr:weight} is correct.
\end{lemma}
\begin{proof}
	By definition, for every vertex $v$, any solution must contain~$v$ or at least $\deg(v) + w_v - r = \deg(v) + (w_v - 1) - (r - 1)$ of its neighbors.
\end{proof}

\begin{proposition} %
	\label{thm:bdd_fen}
	\BDD{} admits a kernel of size $\bigO(\fen^2)$.
\end{proposition}
\begin{proof}
	First, we show that after applying all our reduction rules, the weight of every vertex is at most $\bigO(\fen)$.
	Since \Cref{rr:weight} has been applied exhaustively, there exists a vertex $v \in V$ with $w_v = 0$.
	If $r \ge \deg(v) + w_v = \deg(v)$, then \Cref{rr:wbdd_deg_0} was not exhaustively applied.
	Thus $r < \deg(v)$, which by \Cref{lm:v3} is in $\bigO(\fen)$.
	We have applied \Cref{rr:wbdd_high_weight_rule}, and hence for each vertex $v \in V(G)$, $w_{v} \le r \in \bigO(\fen)$.
	An instance $(G, k, r, w)$ of \wbdd{} is equivalent to an instance $(G', k, r)$ of \BDD{}, where $G'$ is a graph obtained by adding to $w_v$ neighbors to every vertex $v$.
	Thus, we obtain a kernel of size $\bigO(\fen^2)$.
\end{proof}

It is straightforward to adapt our algorithm to \BFVD{} with $i \ge 2$:

\begin{rrule}
	\label{rr:fen:bfvd-1}
	If $v$ is a vertex with $\deg(v) = 1$, then delete $v$.
\end{rrule}

\Cref{rr:fen:bfvd-1} is correct since a degree-one vertex is not part of any biclique when $i \ge 2$.

\begin{rrule}
	\label{rr:fen:bfvd-2}
	If $(v_1, v_2, v_3, v_4, v_5)$ is a path on five vertices with $\deg(v_i) = 2$ for each $i \in \{ 2, 3, 4 \}$, then delete~$v_3$.
\end{rrule}

\begin{lemma}
	\Cref{rr:fen:bfvd-2} is correct.
\end{lemma}
\begin{proof}
	It suffices to show that $v_3$ is not part of any biclique $K_{i,j}$ with $i \ge 2$.
	Suppose that $v_3$ is part of a biclique $(S, T)$ with $|S| \ge 2$ and $|T| \ge 2$.
	Without loss of generality, assume that $v_3 \in S$.
	Then the only two neighbors of $v_3$, namely, $v_2$ and $v_4$ must be contained in $T$.
	Note, however, that $N(v_2) \cap N(v_4) = \{ v_1, v_3 \} \cap \{ v_3, v_5 \} = \{ v_3 \}$, implying that $|S| = 1$, a contradiction.
\end{proof}

One can apply \Cref{rr:fen:bfvd-1,rr:fen:bfvd-2} exhaustively in linear time.
By \Cref{lm:v3}, we have the following:

\begin{proposition}
	\label{prop:bfvd-fen}
	\BFVD{} admits a kernel of size $\bigO(\fen)$ for $i \ge 2$.
\end{proposition}

\Cref{thm:bdd-fen-kernel} follows from \Cref{thm:bdd_fen,prop:bfvd-fen}.

\section{Conclusion}

In this work, we introduced the \textsc{Biclique Free Vertex Deletion} \textsc{(BFVD)} problem and investigated its parameterized complexity with respect to structural parameters.
We showed that \BFVD{} is FPT for $d + k$, where $d$ is the degeneracy and $k$ is the solution size.
This implies fixed-parameter tractability for the feedback vertex number $\fvn$ when $i \ge 2$.
One natural question is whether the problem also admits a polynomial kernel for $\fvn$.
Recently, it was shown that all maximal bicliques can be enumerated efficiently on graphs of bounded weak closure~\citep{DBLP:journals/algorithmica/KoanaKS23a}, which is a superclass of degenerate graphs.
Is \BFVD{} also FPT when parameterized by the weak closure and the solution size?

\bibliographystyle{abbrvnat}
\bibliography{bibliography}

\end{document}